\documentclass{article}
\usepackage{graphicx} 
\usepackage{placeins} 
\usepackage[ruled]{algorithm2e}
\usepackage{amsthm, amsmath, amssymb, mathtools}

\usepackage{tikz, url}
\usetikzlibrary{cd}
\usetikzlibrary{shapes,positioning,calc}
\tikzstyle{hybp}=[rectangle,draw,minimum size=2mm,inner sep=0.1pt]
\tikzstyle{trep}=[circle,draw,minimum size=2mm,inner sep=0.1pt]
\tikzstyle{tre}=[circle,draw,minimum size=1.5mm]
\tikzstyle{small}=[circle,draw,inner sep=0pt, minimum size=1mm]

\usepackage[
  style=numeric,
  sorting=nty,
  natbib=true,
  isbn=false,
  doi=false,
  url=false,
  eprint=false,
  giveninits=true,
  sortcites
  ]{biblatex}

\usepackage{hyperref}

\DeclareSourcemap{
  \maps[datatype=bibtex, overwrite]{
    \map{
      \step[fieldset=address, null]
      \step[fieldset=location, null]
    }
  }
}

\AtEveryBibitem{%
  \ifentrytype{techreport}{%
  }{%
    \clearfield{url}%
    \clearfield{urldate}%
  }%
}

\AtEveryBibitem{\clearfield{month}}
\AtEveryBibitem{\clearfield{day}}

\newcommand{\Ext}{\mathrm{LinExt}}
\newcommand{\Seq}{\mathrm{Seq}}
\newcommand{\CRSeq}{\mathrm{CR\text-Seq}}
\newcommand{\CR}{\mathrm{CR}}
\newcommand{\op}{\mathrm{op}}
\newcommand{\tw}{\mathrm{tw}}
\newcommand{\CC}{\mathcal{C}}

\renewcommand{\int}{\mathring}
\renewcommand\leq\leqslant
\renewcommand\geq\geqslant

\newcommand{\degin}{\mathrm{deg}_{\mathrm{in}}}
\newcommand{\degout}{\mathrm{deg}_{\mathrm{out}}}

\def\quot#1#2{#1/_{{\displaystyle \!#2}}}

\newcommand{\precdot}{\prec\mathrel{\mkern-5mu}\mathrel{\cdot}}

\theoremstyle{plain}
\newtheorem{thm}{Theorem}
\newtheorem{lem}[thm]{Lemma}
\newtheorem{cor}[thm]{Corollary}
\newtheorem{prop}[thm]{Proposition}

\theoremstyle{definition}

\newtheorem{rem}{Remark}

\title{Counting cherry reduction sequences is counting linear extensions (in phylogenetic tree-child networks)}

\author{%
    Tom\'as M. Coronado\textsuperscript{a,1},
    Joan Carles Pons\textsuperscript{a,1}, Gabriel Riera\textsuperscript{a,1,*}
     \\[1em]
    \textsuperscript{a} Department of Mathematics and Computer Science, \\
    Universitat de les Illes Balears, Spain \\[1em]
    \textsuperscript{1} All authors have contributed equally \\[1em]
    \textsuperscript{*} Corresponding author: gabriel.riera@uib.es
}

\date{March 2024}

\begin{document}

\maketitle

\section*{Abstract}

Orchard and tree-child networks share an important property with phylogenetic trees: they can be completely reduced to a single node by iteratively deleting cherries and reticulated cherries. As it is the case with phylogenetic trees, the number of ways in which this can be done gives information about the topology of the network. Here, we show that the problem of computing this number in tree-child networks is akin to that of finding the number of linear extensions of the poset induced by each network, and give an algorithm based on this reduction whose complexity is bounded in terms of the level of the network.

\section{Introduction}

Phylogenetics is the discipline concerned with the study of the evolutionary relationships among taxa that are supposed to evolve in a temporal series of contingent events~\cite{bonnin2019situer}, mainly organisms, genes and languages.
The primary goal is to understand and reconstruct those  relationships, elucidating the evolutionary patterns that have led to their current diversity.

Most of the fundamental goals of phylogenetics can be formulated in mathematical language and as mathematical and computational problems to be solved. In fact, ever since the time of Lamarck and Darwin, the representation of evolution itself has been historically done in terms of graphs.
Traditionally, these graphs have been rooted trees; that is: (directed) rooted acyclic graphs, without nodes allowing multiple incoming edges. The distinguished node (the root) corresponds to the most recent common ancestor of the taxa under consideration;  the leaves represent extant taxa, and interior nodes represent (only) speciation events.

However, the evolutionary history of organisms is often more complex than a simple tree-like structure. Reticulate events, such as hybridization, horizontal gene transfer, and recombination, are better represented in networks rather than in trees.  Phylogenetic networks generalize  phylogenetic trees allowing the possibility for nodes to have multiple (in the binary setting, exactly two) incoming edges. Those nodes represent the reticulation events; see Subsection~\ref{subsec:networks} for formal definitions.

A particularity of (binary) phylogenetic trees is that they can all be generated by means of iteratively adding cherries (pairs of leaves; see Subsection~\ref{subsec:orchard} for a formal definition) to a tree with only one leaf~\cite[p. 28]{semple2003phylogenetics}; but this is tantamount to saying that any particular tree can be recursively reduced to a single leaf. These reductions have an evolutionary meaning, since they are used to distill or simplify the evolutionary information contained in a phylogenetic tree into a more manageable form.  However, it is not possible, given a general phylogenetic network, to reduce it to a single node. In 2019, a new class of phylogenetic networks was introduced whose aim was to mimic precisely this property: the class of orchard, or cherry-picking, phylogenetic networks~\cite{erdos2019orchard, Janssen2021}.  Roughly speaking, these are defined as phylogenetic networks that can be reduced to a single leaf by an iterative process of reductions consisting in (1) simplifying a cherry, or (2) simplifying a reticulated cherry; see Subsection~\ref{subsec:orchard} for formal definitions. Orchard networks were initially introduced for their computational benefits: for instance, providing a polynomial-time algorithm for their reconstruction from its ``ancestral profile'' in~\cite{erdos2019orchard}, or providing a linear-time algorithm for the Network Containment problem --- a generalization of the Tree Containment problem --- for inputs of tree-child networks (a subclass of orchard networks), which is NP-complete for general phylogenetic networks, in~\cite{Janssen2021}. Moreover, orchard networks are integrated in the set of biologically relevant phylogenetic networks  (according to~\cite{van2022orchard} and~\cite{kong2022classes}).

For any given orchard network, there may be  multiple ways of reducing an orchard network to a single leaf, and this number gives information about the topology of the network. In~\cite{erdos2019orchard} the question about the feasibility of counting the number of these reductions (called cherry reduction sequences) was posed. In this paper, we study this question in the context of tree-child networks~\cite{Cardona2009}, which are a subclass of orchard networks, in fact one of the most popular classes of phylogenetic networks. They are extensively studied both for their mathematical and their computational applications; see~\cite{van2010locating,bordewich2016determining,van2014trinets}, to name a few.

In this paper, we study this problem by reducing it to that of counting linear extensions~\cite{kangas2016counting}. Given a poset $(X,\preceq)$, a linear extension is a bijection $\pi:\{1,\ldots, |X|\}\to X$ such that $x\preceq y$ implies $\pi^{-1}(x)\leq\pi^{-1}(y)$; i.e., a total order $\leq$ that is compatible with $\preceq$. Then, an algorithm is proposed whose complexity is studied in terms of the level of the input tree-child network.

The paper is organized as follows. First, we provide some needed background and notation from graph theory, phylogenetics and combinatorics in Section~\ref{sec:preliminaries}. Section~\ref{sec:lin-ext} contains the main theoretical result of this paper, which is then used at the start of Section~\ref{sec:alg} to provide an algorithm built on top of the work of~\cite{kangas2016counting}. In the same section, upper and lower bounds of the complexity of our algorithm relating the tree-width parameter and the level of a phylogenetic network are discussed in the context of determining the complexity of the algorithm without having to study its tree-width. The paper ends with a Conclusions section.

\section{Preliminaries} \label{sec:preliminaries}

\subsection{Graph theory}\label{ssec:graph-theory}

In this paper we will need some graph theory definitions that are not always used in the study of phylogenetic networks. Hence, we have opted to present them here in an independent fashion. First of all, for a directed or undirected graph $G$ we denote its sets of nodes and edges by $V(N)$ and $E(N)$, respectively, and for any directed graph $G$, we shall denote its underlying (undirected) graph by $U(G)$. In order to ease the proofs, we will not allow the existence of loops (that is, edges that begin and end in the same node).

A subgraph of an undirected graph is \emph{biconnected} when it is connected and it remains connected after removing any node from it and all arcs incident to that node. A subgraph of a directed graph $G$ is \emph{biconnected} when it is so in $U(G)$.

Given a directed graph $G$, we say that it is \emph{weakly connected} if its underlying graph, $U(G)$, is connected. If $U(G)$ is not connected, then each of the connected components of $U(G)$ is a \emph{weakly connected component} of $G$. In the same fashion, we say that a directed subgraph is \emph{strongly connected} if it is connected as a directed graph; moreover, a \emph{strong path} is a path connecting one to the other as a directed graph.

Let $G$ be a graph and $e = uv$ an edge, for some $u,v\in V(G)$. The \emph{contraction} of $e$ in $G$, denoted by $\quot G e$, is the result of subtracting $e$ and identifying $u$ and $v$ in a node $w$ that inherits all the adjacencies of both $u$ and $v$. Formally, $\quot G e$ is the \emph{quotient} $\quot G\sim$, where $\sim$ identifies $u$ and $v$ but leaves every other node alone (and then, since we do not allow loops, we remove the edge between $[u]$ and $[v]$). When applying any edge contraction to a directed graph, we always assume the result to be undirected, as it may produce conflicting edge directions.

A \emph{minor} of an undirected graph $G$ is the result of repeatedly applying the process of contraction to a subgraph of $G$~\cite{diestel2000graph}. In what follows, we shall view minors as quotients of these subgraphs, and therefore each of their nodes as an equivalence class. Notice that the nodes in each of these classes form a connected subgraph of $G$. We shall often refer to a minor of a directed graph, in which case we are refering to a minor of its underlying graph.

We denote the \emph{tree-width}~\cite{diestel2000graph} of a graph $G$ by $\tw(G)$; and if $G$ is a directed graph, $\tw(G)$ will denote the tree-width of its underlying graph $U(G)$. A precise definition of the tree-width of a graph can be found in~\cite{diestel2000graph}, but to our purposes it suffices to know that if $H$ is a minor of $G$, then $\tw(H)\leq\tw(G)$ (Lemma 12.4.1 in~\cite{diestel2000graph}) and that the tree-width of a clique $K_n$ is $n-1$~\cite{robertson1986treewidth}.

\subsection{Partial orders and reachability}

For any set, a \emph{partial order} over it is a relation $\preceq$ that is reflexive, antisymmetric and transitive. A set $X$ with such a relation is called a \emph{partially ordered set}, or \emph{poset} and given as $(X,\preceq)$ (although we often omit this and write $X$). If there exists $x\in X$ such that there is no $y\in X\setminus\{x\}$ such that $y\preceq x$, we say that $x$ is \emph{minimal}. Analogously, if there exists $x\in X$ such that there is no $y\in X\setminus\{x\}$ and $x\preceq y$, we say that $x$ is \emph{maximal}.

Given a poset $(X,\preceq)$ and $x,y \in X$, if $x \neq y$ and $x \preceq y$ then $x \prec y$. Moreover, we say that $y$ \emph{covers} $x$ and denote it by $x \precdot y$  if $x \prec y$ and there exists no other element $z \in X$ such that $x \prec z \prec y$. For any poset $X$, its \emph{cover graph} (sometimes referred to as \emph{Hasse diagram}) is a directed graph whose set of nodes is $X$ and its set of edges is $\{xy\in X\times X : x\precdot y\}$. We denote the cover graph of $X$ as $\mathcal C(X)$.

Given any poset $X$, we can define its \emph{order-dual} poset $(X, \succeq)$, also denoted by $X^\op$, which is a poset with the same underlying set endowed with a relation $\succ$ defined by the rule $x\succeq y$ if, and only if, $y\preceq x$ for all $x,y\in X$.

A \emph{linear extension} of a poset $(X, \preceq)$ is a bijection $\pi:\{1,\ldots, |X|\}\to X$ such that for all $x,y\in X$, $x\preceq y$ implies $\pi^{-1}(x)\leq \pi^{-1}(y)$. The set of all linear extensions of $X$ is denoted by $\Ext(X)$.

For any directed acyclic graph $N$ (henceforth, a DAG), any subset $X\subseteq V(N)$ is endowed with a partial order by the structure of the graph as follows: a node $v$ is said to be \emph{reachable from $u$} if there exists a path $u\leadsto v$ within $N$. Then relation of reachability ($\leadsto$) induces a partial order over any subset $X\subseteq V(N)$: for any $u,v\in V(N)$, $u\preceq v$ if, and only if, $u\leadsto v$.

\subsection{Phylogenetic networks}\label{subsec:networks}

Let $\Sigma$ be a finite set of labels. By a \emph{phylogenetic network} on $\Sigma$ we understand a rooted directed acyclic graph (rDAG) where each node of in-degree $\geq 2$ has out-degree exactly 1 and whose \emph{leaves} (i.e., its nodes of out-degree 0) are bijectively labeled in $\Sigma$~\cite{Huson}.  A \emph{phylogenetic tree} is simply a phylogenetic network without nodes of in-degree $\geq 2$. The definition of phylogenetic tree and network also forbids, for reconstructibility reasons, the existence of \emph{elementary nodes}, that is, of nodes of in-degree and out-degree both equal to 1.

Let  $N$ be a phylogenetic network on $\Sigma$. We shall denote its \emph{root} (i.e., its only node of in-degree 0) by $\rho$ and we shall always identify its leaves with their corresponding labels.
Given two nodes $u,v$ in $N$, we say that $v$ is a \emph{child} of $u$, and also that $u$ is a \emph{parent} of $v$, when $uv\in E(N)$. A node in $N$ is of \emph{tree type}, or a \emph{tree node}, when its in-degree is $\leq 1$ (thus including the root), and  a \emph{reticulation} when its in-degree is $\geq 2$ (and hence, its out-degree is 1). For any rDAG $N$, we define $T(N)$ as the set of tree nodes of $N$, i.e., the set of nodes of $N$ with in-degree at most $1$; conversely, we denote by $R(N)$ the set of its reticulation nodes.

In the context of phylogenetic networks, maximal biconnected components are often referred to as \emph{blobs}. Moreover, every blob of a phylogenetic network has one, and only one, node that is ancestor of all nodes in the blob. The \emph{level} of a phylogenetic network is the maximum number of reticulations among its blobs. Notice that each blob has at least one reticulation. A level-$k$ network is then a phylogenetic network where each of its blobs has at least $k$ reticulations.

A phylogenetic network is \emph{tree-child}~\cite{cardona2008treechild} if all its internal nodes (i.e. its nodes of out-degree $\geq 1$) have at least one child of tree type. In this paper, we restrict ourselves to \emph{binary} phylogenetic networks --- this is, networks where all internal nodes except for the root have either in-degree 1 and out-degree 2 or in-degree 2 and out-degree 1. The root will always have in-degree 0, and out-degree 2 if it is not itself a leaf. We denote that a rDAG $N$ has only one node (i.e. the root is a leaf) by $N \cong \star$.

\subsection{Orchard networks and cherry reductions}\label{subsec:orchard}

Given a phylogenetic network $N$, we denote by $\int N$ the rDAG obtained from $N$ by removing all its leaves and the edges incident to them. If $N \cong \star$, then $\int N = \emptyset$ (as the root has out-degree $0$).

Let $N$ be a phylogenetic network. A \emph{cherry} of $N$ is a subgraph of $N$ comprised of two leaves and their parent node. Analogously, \emph{reticulated cherry} of $N$ is a subgraph of $N$ formed by two leaves, a reticulation node and an internal tree node in such a way that the reticulation node is the parent of one of those leaves and the child of the internal tree node, which is in turn the parent of the other leaf (see Figure~\ref{fig:cherry-red}). A tree node is \emph{terminal} if it is the root of a cherry or a reticulated cherry.

Let $N$ be a phylogenetic network and $u \in T(\int N)$ be a terminal node. The \emph{cherry reduction} rooted on $u$, denoted by $\CR(N, u)$, is the network resulting of removing all the descendant edges and strict descendant tree nodes of $u$ from $N$.
We emphasize that we only remove all the strict descendant \emph{tree} nodes of $u$. As we can see in Figure~\ref{fig:cherry-red}, the reticulation in the right-side image is not removed, but becomes a new leaf instead.
Notice that
$\CR(N, u)$ is a subnetwork of $N$ in the sense that all the nodes and edges in $\CR(N,u)$ are nodes and edges of $N$.

An \emph{orchard network}~\cite{erdos2019orchard} is a network $N$ that can be completely reduced by the repeated application of cherry reductions; i.e., such that there exists a sequence of subnetworks of $N$, namely $(N_0,\, N_1,\, \ldots,\, N_s)$ such that $N_0 = N$, each $N_{i+1}$ is the result of a cherry reduction applied to $N_i$, and $N_s \cong \star$. Notice that any two sequences of cherry reductions of a given orchard network have the same length.
We will denote the set of all such sequences of cherry reductions of a given network $N$ by $\CRSeq(N)$. It will be important to remember that all tree-child networks are orchard networks~\cite{erdos2019orchard}.

\begin{figure}
  \centering
  \begin{tikzpicture}[thick, >=stealth, scale=0.65]
    \node (r1) at (0,0) {$\vdots$};

    \path (r1)
      ++(0,-2)  node[trep, label=above right:\scriptsize $u$] (u1) {}
      ++(-1,-1) node[trep, dashed] (x1) {}
      ++(2,0)   node[trep, dashed] (y1) {};

    \draw[->] (r1) -- (u1) {};
    \draw[->, dashed] (u1) -- (x1) {};
    \draw[->, dashed] (u1) -- (y1) {};

    \node (r2) at (8,0) {$\vdots$};

    \path (r2)
      ++(-2,-2)  node (v2) {$\cdots$}
      ++(2,0)    node[trep, label=above right:\scriptsize $u$] (u2) {}
      ++(-1,-1)  node[hybp] (H2) {}
      ++(2,0)    node[trep, dashed] (y2) {}
      ++(-2,-1)  node[trep, dashed] (x2) {};

    \draw[->] (r2) -- (u2) {};
    \draw[->] (v2) -- (H2) {};
    \draw[->, dashed] (u2) -- (H2) {};
    \draw[->, dashed] (H2) -- (x2) {};
    \draw[->, dashed] (u2) -- (y2) {};
  \end{tikzpicture}
  \caption{The two different cherry reductions: a regular cherry (left) and a reticulated cherry (right).}
  \label{fig:cherry-red}
\end{figure}
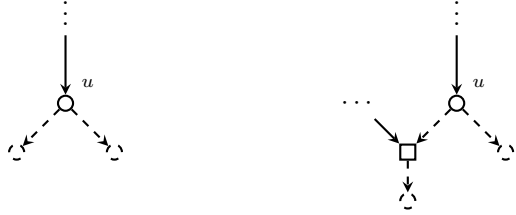

\section{Cherry reductions and linear extensions}\label{sec:lin-ext}

The aim of this section is to prove that there are as many cherry reduction sequences as there are linear extensions of the poset of internal tree nodes induced by a tree-child network (Theorem~\ref{thm:bij}). We shall begin by showing how the set of internal tree nodes is changed each time a cherry-reduction is performed.

\begin{lem} \label{lem:cr-inj}
  Let $N$ be an orchard network, $u\in T(\int N)$ a terminal node, and $N' = \CR(N, u)$. Then $T(\int N) = T(\int N') \sqcup \{u\}$ and $u \in L(N')$.
\end{lem}
\begin{proof}
    That $u\in L(N')$ springs from the fact that $N'$ is obtained from $N$ by removing all strict descendants of $u$ from $N$; thus, $\degout^{N'}(u) = 0$, and this is the definition of a leaf. Therefore, $u\notin\int N'$.

    On the other hand, $u\in T(\int N)$, by definition of $\CR$.
\end{proof}

\begin{rem}
  In particular, $\CR(N, -)$ is injective.
\end{rem}

Remember that, for any set $X$, $\mathfrak S_X$ represents the set of permutations of $X$, i.e. bijective functions $\{1,\ldots,|X|\} \to X$.

\begin{lem}\label{lem:sigma}
  Let $N$ be an orchard network. The assignment
  $$
    \sigma : \CRSeq(N) \to \mathfrak S
  $$
  that maps each $\CR$-sequence $(N_0,\ldots,N_s)$ of $N$ to the permutation $(u_0\; \cdots\; u_{s-1})$, where $N_{i+1} = \CR(N_i, u_i)$, is well-defined and injective.
\end{lem}
\begin{proof}
  For the first part we need to prove that this construction yields a permutation of $T(\int N)$.
  Let $(N_0,\ldots,N_s)$ be a $\CR$-sequence and $u_0,\ldots,u_{s-1}$ the nodes such that $N_{i+1} = \CR(N_i,u_i)$. By Lemma~\ref{lem:cr-inj}, these are uniquely determined by $(N_0,\ldots,N_s)$. In addition, since $T(\int N) \supseteq T(\int N_1) \supseteq \cdots \supseteq T(\int N_s)$ and each $u_i \in T(\int N_i)$, we have $u_i \in T(\int N)$. These are also necessarily different from each other as $u_i \in T(\int N_i) \setminus T(\int N_{i+1})$. Since $\lvert T(\int N) \rvert = s$, it follows that $(u_0\; \cdots\; u_{s-1})$ is indeed a permutation of $T(\int N)$.

  Finally, we must verify that $\sigma$ is injective. Let $S = (N_0,\ldots,N_s)$ and $S' = (N_0',\ldots,N_s')$ be two $\CR$-sequences for $N$. Note that both sequences must have the same length, $\lvert T(\int N) \rvert$. Assume that $\sigma(S) = (u_0\; \cdots u_{s-1}) = \sigma(S')$. Then $N = N_0 = N_0'$ and at each step we have $N_{i+1} = \CR(N_i, u_i) = N_{i+1}'$, which provides the proof by induction.
\end{proof}

The next lemma provides a useful and succinct relationship between the terminal nodes in $T(\int N)$ and the maximal nodes in its associated poset if $N$ is a tree-child network.

\begin{lem}\label{lem:ter=min}
  For any tree-child network $N$, for any $u\in T(\int N)$, $u$ is terminal if, and only if, $u$ is maximal in $T(\int N)$.
\end{lem}
\begin{proof}
  If $u$ is terminal, then either it is the root of a cherry or that of a reticulated cherry. In both cases, it does not have any strict descendant that is internal tree node, and so it is maximal in $T(\int N)$.

  On the other hand, if $u$ is maximal in $T(\int N)$, then it has no strict descendant node that is both internal and a tree node. But since $N$ is a tree-child network, at least one of its two children must be a leaf. Two cases arise: if the other child is also a leaf, then it is the root of a cherry. If the other child is a reticulation, then by the maximality of $u$ among internal tree nodes and the fact that $N$ is tree-child, the only child of this reticulation must be a leaf. In either case, $u$ is a terminal node.
\end{proof}

\begin{rem}\label{rem:not-valid-orchards}
    Lemma~\ref{lem:ter=min} is actually false for general orchard networks. Indeed, consider the following example:
    
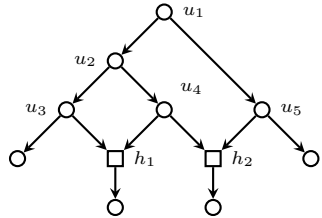
\begin{center}
        \begin{tikzpicture}[scale=0.65, thick, >=stealth]
            \draw(1,1) node[trep, label=right:\scriptsize $u_1$] (r) {};

            \draw(0,0) node[trep, label=left:\scriptsize $u_2$] (1) {};

            \draw(3,-1) node[trep, label=right:\scriptsize $u_5$] (2) {};

            \draw(1,-1) node[trep, label=above right:\scriptsize $u_4$] (3) {};
            \draw(-1,-1) node[trep, label=left:\scriptsize $u_3$] (4) {};

            \draw(-2,-2) node[trep] (5) {};
            \draw (5) node {};
            \draw(0,-2) node[hybp, label=right:\scriptsize $h_1$] (6) {};
            \draw(2,-2) node[hybp, label=right:\scriptsize $h_2$] (7) {};
            \draw(4,-2) node[trep] (8) {};

            \draw(0,-3) node[trep] (9) {};
            \draw(2,-3) node[trep] (10) {};

            \draw[->](r)--(1);
            \draw[->](r)--(2);
            \draw[->](1)--(3);
            \draw[->](1)--(4);
            \draw[->](4)--(5);
            \draw[->](4)--(6);
            \draw[->](3)--(6);
            \draw[->](3)--(7);
            \draw[->](2)--(7);
            \draw[->](2)--(8);
            \draw[->](6)--(9);
            \draw[->](7)--(10);

        \end{tikzpicture}
    \end{center}
    This network is obviously orchard: $(u_3\ u_5\ u_4\ u_2\ u_1)$ gives a cherry reduction sequence. However, even though $u_4$ is a maximal internal tree node, there is no cherry reduction sequence that begins by it.

    This same example negates the validity of the statement of Theorem~\ref{thm:bij} below for general orchard networks.
\end{rem}

In what follows, to ease the proofs, we shall proceed with the order-dual of the poset $T(\int N)$. Indeed, since the first reduced nodes are the last nodes if we count from the root, as it is standard practice, we shall  reverse that order in the proofs below.

\begin{lem}
  Let $N$ be an orchard network. The mapping $\sigma$ in Lemma~\ref{lem:sigma} can be restricted to $\tilde\sigma : \CRSeq(N) \to \Ext(T(\int N)^\op)$, this is, the set of linear extensions of the order-dual of the poset $T(\int N)$.
\end{lem}
\begin{proof}
  Let $S = (N_0,\ldots,N_s) \in \CRSeq(N)$ and $\pi = \tilde\sigma(S)$. According to the previous observation, we must prove that $\pi$ is a linear extension of the dual of partial order on $T(\int N)$ induced by reachability relation of $N$. In other words, that for any two $u, v \in T(\int N)$ such that $v \rightsquigarrow u$ in $N$, then $\pi^{-1}(u) \leq \pi^{-1}(v)$. By definition, $\pi^{-1}(u)$ is the largest index $i$ such that $u \in T(\int N_i)$ (resp. $j$ for $v$). Assume that $i > j$, so that we have the following situation:
  $$
  \begin{array}{ccccccccc}
    \cdots \supseteq           & T(\int N_j)     &
    \supseteq                  & T(\int N_{j+1}) &
    \supseteq \cdots \supseteq & T(\int N_i)     &
    \supseteq                  & T(\int N_{i+1}) &
    \supseteq \cdots \\
    & \rotatebox{90}{$v \in$}    \; \rotatebox{90}{$u \in$} &
    & \rotatebox{90}{$v \notin$} \; \rotatebox{90}{$u \in$} &
    & \rotatebox{90}{$v \notin$} \; \rotatebox{90}{$u \in$} &
    & \rotatebox{90}{$v \notin$} \; \rotatebox{90}{$u \notin$} &
  \end{array}
  $$
  We can observe that, at $N_{j+1}$, $v$ is a leaf and $u$ is an internal node, which contradicts the assumption that $v \rightsquigarrow u$ in $N$.
\end{proof}

\begin{rem}
  Since $\sigma = \iota \circ \tilde\sigma$, where $\iota : \Ext(T(\int N)^\op) \hookrightarrow \mathfrak S $ is the natural inclusion, then $\tilde\sigma$ must also be injective.
\end{rem}

We shall now present Theorem~\ref{thm:bij}, the main result in this section, stating that counting cherry reduction sequences of a tree-child network is equivalent to counting the set of linear extensions of the poset $T(\int N)$.

\begin{thm}\label{thm:bij}
    Let $N$ be a tree-child network. Then, there is a bijection between the set of cherry reduction sequences of $N$ and the set of linear extensions of the poset $T(\int N)$.
\end{thm}
\begin{proof}
    It will be enough to prove that $\tilde\sigma$ is surjective, and thus, by Lemma~\ref{lem:sigma}, a bijection. Indeed, consider the following commutative diagram:
    \begin{center}
      \begin{tikzcd}
        \CRSeq(N)
          \arrow[rd, "\cong", "\tilde\sigma"']
          \arrow[rr, "\sigma", tail]
        & &
        \mathfrak S  \\
        &
        \Ext(T(\int N)^\op)
          \arrow[ru, "\iota"', hook]
          \arrow[d, "\cong", "\mathrm{rev}"']
        & \\
        & \Ext(T(\int N))
        &
      \end{tikzcd}
    \end{center}
    Notice that, in this diagram, $\mathrm{rev}$ is simply the bijection that simply reverts the permutations. We can then see that $\mathrm{rev}\circ\tilde\sigma$ provides such bijection.

    Let us prove that $\tilde\sigma$ is surjective. Let $\pi = (u_0\ u_{1} \cdots u_{s-1})$, $u_{s-1} = \rho$, be a linear extension of the order-dual of the reachability order in $N$. We want to show that
    $$
      (N_0,\, N_1 = \CR(N_0,u_0),\, N_2 = \CR(N_1,u_{1}),\,\ldots,\, N_s = \CR(N_{s-1},u_{s-1}))
    $$
    is a cherry reduction sequence for $N = N_0$.

    We shall proceed by induction over $i\in\{0,\ldots, s-1\}$. If $i=0$, then $\CR(N_{0}, u_{0})$ is $N_1$. Since $\pi$ is a linear extension, then $u_0$ is maximal in $N$, and thus by Lemma~\ref{lem:ter=min}, a terminal node in $N$. Furthermore, $N_1$ is still a tree-child network and $T(\int N_1) = T(\int N_0)\setminus\{u_0\}$ (by Lemma~\ref{lem:cr-inj}).

    Now assume that the result holds up to $i\in\{0,\ldots,s-1\}$; that is, that $u_{j}$ is a terminal node of $N_j$ and that $\CR(N_j, u_{j})$ is a tree-child network for any $j \leq i$. We shall now prove that it holds for $i+1$. Now, by the induction hypothesis, $N_{i+1} = \CR(N_{i}, u_{i})$ is a tree-child network. Furthermore, $u_{i+1}$ is an internal tree node of $N_{i+1}$: indeed, since $T(\int N_{i+1}) = T(\int N_i)\setminus\{u_{i}\}$, and $u_{i+1}$ was an internal tree node of $N_i$. Now, clearly, $u_{i+1}$ is a maximal element of the poset $T(\int N_{i+1})$, because $\pi|_{T(\int N_{i+1})}$ begins by $u_{i+1}$. Therefore, it is a terminal node of $N_{i+1}$, and $\CR(N_{i+1}, u_{i+1})$ is a tree-child network whose set of internal nodes is $T(\int N_{i+1})\setminus\{u_{i+1}\}$ (Lemma~\ref{lem:cr-inj}).
\end{proof}

\section{Algorithms and complexity} \label{sec:alg}

Now that we have established that counting cherry reduction sequences in a tree-child network is the same as counting linear extensions of the induced poset of its internal tree nodes, we shall work towards giving an algorithm to compute this number.

The problem of counting linear extensions of a poset is \#P-complete in general~\cite{brightwell1991sharpcomplete}, and several algorithms have been proposed to count the number of linear extensions of posets under different restrictions~\cite{atkinson1990computing, felsner2015linear, habib1987some, kangas2016counting, li2005computing, mohring1989computationally, peczarski2004new}, or estimating it~\cite{bubley1999faster, dyer1991random}.

The tree-width of a graph has been thoroughly used in order to assess the complexity of algorithms that count the number of linear extensions. For instance, in~\cite{eiben2019counting} it is shown that the problem of counting the number of linear extensions of a poset of $p$ elements parameterized by the tree-width $\tw$ of its underlying cover graph is not fixed-parameter tractable (FPT), or in other words, there is no algorithm with complexity $f(\tw)p^{O(1)}$ for some computable function $f$. On the other hand, the level of a phylogenetic network is a commonly used parameter to asses the complexity of algorithms in phylogenetics. Usually an upper bound of the level of a network is the easier one to assess, and fortunately the tree-width and level are related by Lemma~\ref{lem:janssen}.

\begin{lem}[{\cite[Observation 2]{janssen2019treewidth}}]\label{lem:janssen}
  Let $N$ be a level-$k$ network and $r$ its number of reticulations. Then $\tw(N) \leq k + 1 \leq r + 1$.
\end{lem}

We will follow the lead of~\cite{kangas2016counting}, where the authors present two algorithms to count the number of linear extensions of a poset $P$ with $p$ elements, the first one with complexity $O(2^pp)$, and the second one with complexity $O(p^{\tw+4})$, where $\tw$ is the (undirected) tree-width of the cover graph of $P$.

\begin{thm}[{\cite[Theorem 1]{kangas2016counting}}]\label{thm:kangas}
  Given a poset $P$, the number of linear extensions of $P$ can be computed in time $O(p^{\tw+4})$, where $p = \lvert P \rvert$ and $\tw$ is the tree-width of the cover graph of $P$.
\end{thm}

The problem arises when we want to assess the tree-width of the cover graph of the poset induced over the internal tree nodes of a tree-child network. By Lemma~\ref{lem:cover-graph-minor}, $\mathcal C(T(\int N))$ is a minor of $N$, and thus $\tw(\mathcal C(T(\int N))) \leq \tw(N)$.

Consequently, Theorem~\ref{thm:kangas} provides an algorithm that computes the number of linear extensions of a tree-child network in time $O(|T(\int N)|^{k+5})$, where $k$ is the level of $N$, and the number of internal tree nodes is the number of leaves of the network minus $1$. In other words, this provides an algorithm that runs in time $O(|\Sigma|^{\tw+4}) \subseteq O(|\Sigma|^{k+5})$.

\begin{lem}\label{lem:cover-graph-minor}
    Let $N$ be a tree-child network and $N'$ the subgraph of $N$ induced by removing, from $\int N$, all the reticulation nodes with out-degree $0$ (in $\int N$). The DAG resulting from the contraction of every edge between a reticulation node and its child tree node in $N'$ is isomorphic to the cover graph of $T(\int N)$.
\end{lem}
\begin{proof}
    First, observe that, in this case, the resulting graph after edge contractions is still directed and acyclic. Now let $\varphi:\CC(T(\int N))\to\quot{N'}{\sim}$ defined by the rule $u\mapsto [u]$, where $\quot{N'}{\sim}$ is the (quotient) DAG defined in the statement of this proposition by contracting edges from reticulations to tree nodes.

    Let us first see that $\varphi$ is indeed bijective. To see that the map is surjective, notice that nodes $[u]$ in $\quot{N'}{\sim}$ can be either $\{u\}$ or $\{u, h\}$, where $h$ is a reticulation node. In both cases, $u\in T(\int N)$ because $N'$ does not contain any reticulation node with leaf children and thus the map is surjective. Now, let $u,v$ be two tree nodes; if $[u] = [v]$, then either $\{u\} = \{v\}$ or $\{u, h\} = \{v, h\}$, where $h$ is a reticulation node and thus not a tree node; in both cases, $u = v$.

    Let us now see that $\varphi$ preserves and reflects edges. Let $u,v$ be two internal tree nodes. By definition, we know that $u\precdot  v$ if, and only if, there exists a directed path in $N$ between $u$ and $v$ and there is no tree node $w\notin\{u,v\}$ such that $u \prec  w \prec  v$; or, equivalently, that either $uv\in E(N)$ or $uh, hv\in E(N)$ (since the network is tree-child) for some reticulation node $h$. In any case, this implies that $[u][v]\in E(\quot{N'}{\sim})$. Conversely, if $[u][v] \in E(\quot{N'}{\sim})$, then we have three cases for $u' \in [u], v' \in [v]$ such that $u'v' \in E(N)$:
    \begin{itemize}
      \item If $u' = u$ and $v' = v$, then $u \precdot  v$.
      \item If $u'$ is a reticulation node, then $[u] = \{u, u'\}$ and $v' = v$ (because the network is tree-child). Hence $uu', u'v \in E(N)$ and $u \precdot  v$.
      \item If $v'$ is a reticulation node, then $[v] = \{v, v'\}$ and $u' = u$ (because the network is tree-child). Hence $uv', v'v \in E(N)$ and $u \precdot  v$.
    \end{itemize} \end{proof}

The combination of the algorithm in~\cite{kangas2016counting} and Lemma~\ref{lem:cover-graph-minor} is  provided as Algorithm~\ref{alg:cherry-reds}. It uses Lemma~\ref{lem:cover-graph-minor} to first obtain $\mathcal C(T(\int N))$, and then immeditely returns the result of applying any linear extension-counting algorithm (named \emph{LEcount} as in~\cite{kangas2016counting}) to its adjacency matrix. Clearly, both the first step and obtaining the adjacency matrix have complexity at most $O(\lvert V(N) \rvert^2) = O(\lvert \Sigma \rvert^2)$, hence the algorithm's complexity is as high as the complexity of the chosen \emph{LEcount} implementation.

More precisely, the algorithm assumes that the network is provided as an adjacency list, with nodes numbered from $1$ to $\lvert V(N)\rvert$. As seen in Lemma~\ref{lem:cover-graph-minor}, the cover graph of $T(\int N)$ is obtained by contracting edges from a subgraph of $\int N$. To do this, a mapping (\textit{quot}) is constructed by iterating linearly over the nodes $u$ of $N$: leaves are discarded (\textit{quot[u]} $= -1$), internal nodes mapped to the next available index and reticulations $H$ identified with their tree-child $v$ if $v$ is internal (\textit{quot[H] = quot[v]}). Note that attention has to be paid in order not to map the same node twice, both when visited as an internal node and as the child of a reticulation. Once this mapping is constructed, the cover graph of $T(\int N)$ is easily obtained by applying it to all the edges of the input network and then passed to the linear extension-counting algorithm of choice (in the case of \emph{LEcount}, the input must be given as an adjacency matrix).

\begin{algorithm}
  \SetKwComment{Comment}{/* }{ */}
  \SetKwFor{Function}{function}{do}{end}
  \KwIn{Tree-child phylogenetic network $N$ as an adjacency list, $LEcount$ function}
  \KwOut{Number of cherry reduction sequences of $N$}

  $n \gets \lvert V(N) \rvert$\;

  \Comment{Build $\mathcal C(T(\int N))$ using Lemma~\ref{lem:cover-graph-minor}}
  $c \gets 0$\Comment*{To be $\lvert T(\int N) \rvert$}
  \Comment{First, we build a mapping from $1,\ldots,n$ to $-1,1,\ldots,c$ to remove/merge nodes as a quotient function $V(N) \to T(\int N) \sqcup \{-1\}$}
  $quot \gets [-1,\, \overset{n}\ldots,\, -1]$\;
  \For{$u \in V(N)$}{
    \Comment{$u$ is an internal tree node which has not been mapped yet}
    \If{$\degout^N(u) = 2$ and $quot[u] < 0$}{
      $c \gets c + 1$\;
      $quot[u] \gets c$\;
    }
    \Comment{$u$ is a reticulation (with a tree-child)}
    \If{$\degout^N(u) = 1$}{
      $\{v\} \gets N[u]$\;
      \Comment{Consider $u$ only if its child $v$ is internal}
      \If{$\degout^N(v) > 0$}{
        \Comment{Map $v$ first if not visited yet}
        \If{$quot[v] < 0$}{
          $c \gets c + 1$\;
          $quot[v] \gets c$\;
        }
        $quot[u] \gets quot[v]$
      }
    }
    \Comment{If $\degout^N(u) = 0$, then $u$ is a leaf, skip it}
  }
  $coverGraph \gets [\emptyset,\, \overset{c}\ldots,\, \emptyset]$\;
  \For{$u \in V(N)$}{
    \Comment{$u$ is an internal tree node}
    \If{$\degout^N(u) = 2$}{
      $coverGraph[quot[u]] \gets \{ quot[v] : v \in N[u],\, quot[v] \geq 0 \}$\;
    }
  }

  $M_{\mathcal C} \gets adjacencyMatrix(coverGraph)$\;
  \Return{LEcount($M_{\mathcal C}$)}\;
  \caption{Algorithm for counting cherry reductions using LEcount~\cite{kangas2016counting}.}
  \label{alg:cherry-reds}
\end{algorithm}

In the next section, we shall further study the relationship of tree-child networks with their tree-width. In particular, we shall see that we can build tree-child networks of any tree-width, provided (as seen in Lemma~\ref{lem:janssen}) that the level is unbounded.

\FloatBarrier

\subsection{Tree-child networks of arbitrary tree-width}\label{ssec:arb-tw}

The restriction to tree-child networks does not seem to simplify the complexity of the existing algorithms: we shall always be able to build a tree-child network whose tree-width and level are in a similar order of magnitude.

\begin{prop} \label{prop:treewidth-to-level}
  For every $n \geq 3$, there exists a level-$\frac{(n-1)(n-2)}{2}$ tree-child phylogenetic network whose underlying graph has a minor of tree nodes isomorphic to the clique $K_n$. In other words, it has tree-width at least $n$.
\end{prop}
\FloatBarrier
\begin{proof}
  We explicitly construct said phylogenetic network $N$ as follows: Consider the set of nodes
  \begin{align*}
    V = & \{ u_{ij} : 1 \leq i < j \leq n \}
          \cup \{ h_{ij} : 2 \leq i < j \leq n \} \\
        & \cup \{ v_{ij} : 2 \leq i < j-1,\, j \leq n \}
          \cup \{ x_{ij} : 2 \leq i < j-1,\, j \leq n \} \\
        & \cup \{ z_j : 1 \leq j \leq n \}
  \end{align*}
  and edges defined as
  \begin{align*}
    E_j =&\,\{ u_{j\ell} u_{j(\ell+1)} : j < \ell < n \}
           \cup \{ u_{jn} z_j \} \tag{$j = 1, 2$} \\[2ex]
    E_j =& \bigcup_{i = 2}^{j-2} \{ h_{ij} v_{ij},\, v_{ij} x_{ij},\, v_{ij} h_{(i+1)j} \}  \tag{$j = 3,\, \ldots,\, n-1$} \\
         & \cup \{ h_{(j-1)j} u_{j(j+1)} \} \\
         & \cup \{ u_{j\ell} u_{j(\ell+1)} :  j < \ell < n \}
           \cup \{ u_{jn} z_j \} \\[2ex]
    E_n =& \bigcup_{i = 2}^{n-2} \{ h_{in} v_{in},\, v_{in} x_{in},\, v_{in} h_{(i+1)n} \}
           \cup \{ h_{(n-1)n} z_n \} \\[2ex]
    E =  & \bigcup_{j=1}^n E_n
           \cup \{ u_{1j} h_{2j} : 3 \leq j \leq n \}
           \cup \{ u_{ij} h_{ij} : 2 \leq i < j \leq n \},
  \end{align*}
  which can also be visualized in Figures~\ref{fig:treeewidth-to-level-nodes12},~\ref{fig:treeewidth-to-level-noden} and~\ref{fig:treeewidth-to-level-nodej}. Figure~\ref{fig:treeewidth-to-level-examples} shows two examples of this network, with $n=3,4$.
  \begin{figure}[hpt]
    \centering
    \begin{tikzpicture}[scale=0.5, every node/.style={scale=0.8}, thick, >=stealth]
      \draw(0,0) node[trep, color=red, label=left:${u_{12}}$] (u12) {};

      \path (u12)
        ++( 1, -1) node[trep, color=blue, label=below right:$u_{23}$] (u23)  {}
        ++(-2,  0) node[trep, color=red, label=left:$u_{13}$]      (u13) {}
        ++( 1, -1) node[hybp, label=below right:$h_{23}$] (h23)  {}
        ++(-2,  0) node[trep, color=red, label=left:$u_{14}$]      (u14) {}
        ++( 1, -1) node (hddots) {\rotatebox{80}{$\ddots$}}
        ++(-2,  0) node (uddots) {\rotatebox{80}{\textcolor{red}{$\ddots$}}}
        ++(-1, -1) node[trep, color=red, label=left:$u_{1n}$]      (u1n) {}
        ++( 1, -1) node[hybp, color=orange, label=below right:$h_{2n}$] (h2n) {}
        ++(-2,  0) node[trep, color=red, label=left:$z_1$]         (z1) {};

      \draw[->] (u12) -- (u23); \draw[->, color=red] (u12) -- (u13);
      \draw[->] (u13) -- (h23); \draw[->, color=red] (u13) -- (u14);
      \draw[->] (u14) -- (hddots); \draw[->, color=red] (u14) -- (uddots);
      \draw[->, color=red] (uddots) -- (u1n);
      \draw[->] (u1n) -- (h2n); \draw[->, color=red] (u1n) -- (z1);
    \end{tikzpicture}
    \hspace{1em}
    \begin{tikzpicture}[scale=0.5, every node/.style={scale=0.8}, thick, >=stealth]
      \draw(0, 0) node[trep, color=blue, label=right:${u_{23}}$] (u23) {};

      \path (u23)
        ++(-1, 1) node[trep, color=red, label=above:$u_{12}$] (u12) {};

      \path (u23)
        ++(-1, -1) node[hybp, label=below left:$h_{23}$]    (h23)  {}
        ++( 2,  0) node[trep, color=blue, label=right:$u_{24}$]         (u24) {}
        ++(-1, -1) node[hybp, label=below left:$h_{24}$]    (h24)  {}
        ++( 2,  0) node[trep, color=blue, label=right:$u_{25}$]         (u25) {}
        ++(-1, -1) node (hddots) {$\ddots$}
        ++( 2,  0) node (uddots) {\textcolor{blue}{$\ddots$}}
        ++( 1, -1) node[trep, color=blue, label=right:$u_{2n}$]         (u2n) {}
        ++(-1, -1) node[hybp, color=orange, label=below left:$h_{2n}$]    (h2n) {}
        ++( 2,  0) node[trep, color=blue, label=right:$z_2$]            (z2) {};

      \draw[->] (u12) -- (u23);
      \draw[->] (u23) -- (h23); \draw[->, color=blue] (u23) -- (u24);
      \draw[->] (u24) -- (h24); \draw[->, color=blue] (u24) -- (u25);
      \draw[->] (u25) -- (hddots); \draw[->, color=blue] (u25) -- (uddots);
      \draw[->, color=blue] (uddots) -- (u2n);
      \draw[->] (u2n) -- (h2n); \draw[->, color=blue] (u2n) -- (z2);
    \end{tikzpicture}
    \caption{Nodes adjacent to the paths $u_{12}\cdots u_{1n}z_1$ (red) and $u_{23}\cdots u_{2n}z_2$ (blue).}
    \label{fig:treeewidth-to-level-nodes12}
  \end{figure}
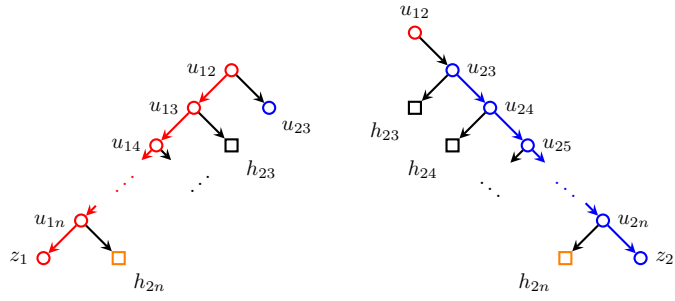
  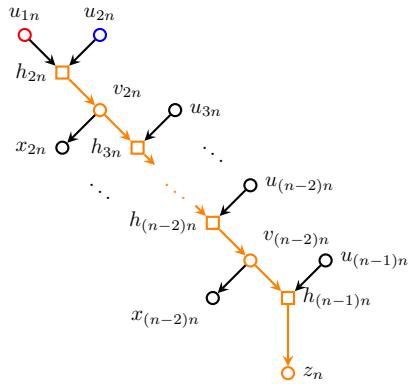
\begin{figure}[hptb]
    \centering
    \begin{tikzpicture}[scale=0.5, every node/.style={scale=0.8}, thick, >=stealth]
      \draw(0,0) node[trep, color=orange, label=right:$z_n$] (zn) {};

      \path (zn)
        ++( 0,  2) node[hybp, color=orange, label=right:$h_{(n-1)n}$]       (hn-1n) {}
        ++(-2,  0) node[trep, label=below left:$x_{(n-2)n}$]  (xn-2n) {}
        ++( 1,  1) node[trep, color=orange, label=above right:$v_{(n-2)n}$] (vn-2n) {}
        ++(-1,  1) node[hybp, color=orange, label=left:$h_{(n-2)n}$]        (hn-2n) {}
        ++(-3,  1) node (xddots) {$\ddots$}
        ++( 2,  0) node (vddots) {\textcolor{orange}{$\ddots$}}
        ++(-1,  1) node[hybp, color=orange, label=left:$h_{3n}$]      (h3n) {}
        ++(-2,  0) node[trep, label=left:$x_{2n}$]            (x2n) {}
        ++( 1,  1) node[trep, color=orange, label=above right:$v_{2n}$]     (v2n) {}
        ++(-1,  1) node[hybp, color=orange, label=left:$h_{2n}$]            (h2n) {}
        ++(-1,  1) node[trep, color=red, label=above:$u_{1n}$]           (u1n) {}
        ++( 2,  0) node[trep, color=blue, label=above:$u_{2n}$]           (u2n) {};

      \path (v2n)    ++( 2,  0) node[trep, label=right:$u_{3n}$] (u3n) {};
      \path (h3n)    ++( 2,  0) node (uddots) {$\ddots$};
      \path (vddots) ++( 2,  0) node[trep, label=right:$u_{(n-2)n}$] (un-2n) {};
      \path (vn-2n)  ++( 2,  0) node[trep, label=right:$u_{(n-1)n}$] (un-1n) {};

      \draw[->] (u1n) -- (h2n); \draw[->] (u2n) -- (h2n);
      \draw[->, color=orange] (h2n) -- (v2n);
      \draw[->] (v2n) -- (x2n);
      \draw[->, color=orange] (v2n) -- (h3n); \draw[->] (u3n) -- (h3n);
      \draw[->, color=orange] (h3n) -- (vddots);
      \draw[->, color=orange](vddots) -- (hn-2n); \draw[->] (un-2n) -- (hn-2n);
      \draw[->, color=orange] (hn-2n) -- (vn-2n);
      \draw[->] (vn-2n) -- (xn-2n);
      \draw[->, color=orange] (vn-2n) -- (hn-1n); \draw[->] (un-1n) -- (hn-1n);
n-1
      \draw[->, color=orange] (hn-1n) -- (zn);
    \end{tikzpicture}
    \caption{Nodes adjacent to the path $h_{2n} \cdots h_{(n-1)n} z_n$ (highlighted in orange).}
    \label{fig:treeewidth-to-level-noden}
  \end{figure}
  \begin{figure}[hpb]
    \centering
    \begin{tikzpicture}[scale=0.5, every node/.style={scale=0.8}, thick, >=stealth]
      \draw(0,0) node[trep, color=magenta, label=right:${u_{j(j+1)}}$] (ujj+1) {};

      \path (ujj+1)
        ++( 0,  2) node[hybp, color=magenta, label=right:$h_{(j-1)j}$]       (hj-1j) {}
        ++(-2,  0) node[trep, label=below left:$x_{(j-2)j}$]  (xj-2j) {}
        ++( 1,  1) node[trep, color=magenta, label=above right:$v_{(j-2)j}$] (vj-2j) {}
        ++(-1,  1) node[hybp, color=magenta, label=left:$h_{(j-2)j}$]        (hj-2j) {}
        ++(-3,  0) node (xddots) {$\ddots$}
        ++( 2,  1) node (vddots) {\textcolor{magenta}{$\ddots$}}
        ++(-1,  1) node[hybp, color=magenta, label=left:$h_{3j}$]        (h3j) {}
        ++(-2,  0) node[trep, label=below left:$x_{2j}$]  (x2j) {}
        ++( 1,  1) node[trep, color=magenta, label=above right:$v_{2j}$] (v2j) {}
        ++(-1,  1) node[hybp, color=magenta, label=left:$h_{2j}$]        (h2j) {}
        ++(-1,  1) node[trep, color=red, label=above:$u_{1j}$]       (u1j) {}
        ++( 2,  0) node[trep, color=blue, label=above:$u_{2j}$]       (u2j) {};

      \path (v2j)    ++( 2,  0) node[trep, label=right:$u_{3j}$] (u3j) {};
      \path (h3j)    ++( 2,  0) node (uddots) {$\ddots$};
      \path (vddots) ++( 2,  0) node[trep, label=right:$u_{(j-2)j}$] (uj-2j) {};
      \path (vj-2j)  ++( 2,  0) node[trep, label=right:$u_{(j-1)j}$] (uj-1j) {};

      \draw[->] (u1j) -- (h2j); \draw[->] (u2j) -- (h2j);
      \draw[->, color=magenta] (h2j) -- (v2j);
      \draw[->] (v2j) -- (x2j);
      \draw[->, color=magenta] (v2j) -- (h3j); \draw[->] (u3j) -- (h3j);
      \draw[->, color=magenta] (h3j) -- (vddots);
      \draw[->, color=magenta](vddots) -- (hj-2j); \draw[->] (uj-2j) -- (hj-2j);
      \draw[->, color=magenta] (hj-2j) -- (vj-2j);
      \draw[->] (vj-2j) -- (xj-2j);
      \draw[->, color=magenta] (vj-2j) -- (hj-1j); \draw[->] (uj-1j) -- (hj-1j);

      \draw[->, color=magenta] (hj-1j) -- (ujj+1);

      \path (ujj+1)
        ++(-1, -1) node[hybp, label=below left:$h_{j(j+1)}$] (hjj+1) {}
        ++( 2,  0) node[trep, color=magenta, label=right:$u_{j(j+2)}$]      (ujj+2) {}
        ++(-1, -1) node[hybp, label=below left:$h_{j(j+2)}$] (hjj+2) {}
        ++( 2,  0) node[trep, color=magenta, label=right:$u_{j(j+3)}$]      (ujj+3) {}
        ++(-1, -1) node (hddots) {$\ddots$}
        ++( 2,  0) node (uddots) {\textcolor{magenta}{$\ddots$}}
        ++( 1, -1) node[trep, color=magenta, label=right:$u_{jn}$]      (ujn) {}
        ++(-1, -1) node[hybp, color=orange, label=below left:$h_{jn}$] (hjn) {}
        ++( 2,  0) node[trep, color=magenta, label=right:$z_j$]         (zj) {};

      \draw[->] (ujj+1) -- (hjj+1); \draw[->, color=magenta] (ujj+1) -- (ujj+2);
      \draw[->] (ujj+2) -- (hjj+2); \draw[->, color=magenta] (ujj+2) -- (ujj+3);
      \draw[->] (ujj+3) -- (hddots); \draw[->, color=magenta] (ujj+3) -- (uddots);
      \draw[->, color=magenta] (uddots) -- (ujn);
      \draw[->] (ujn) -- (hjn); \draw[->, color=magenta] (ujn) -- (zj);
    \end{tikzpicture}
    \caption{Nodes adjacent to the path $h_{2j} \cdots h_{(j-1)j} u_{j(j+1)} \cdots u_{jn} z_j$ (highlighted in magenta) for $j = 3,\ldots,n-1$.}
    \label{fig:treeewidth-to-level-nodej}
  \end{figure}
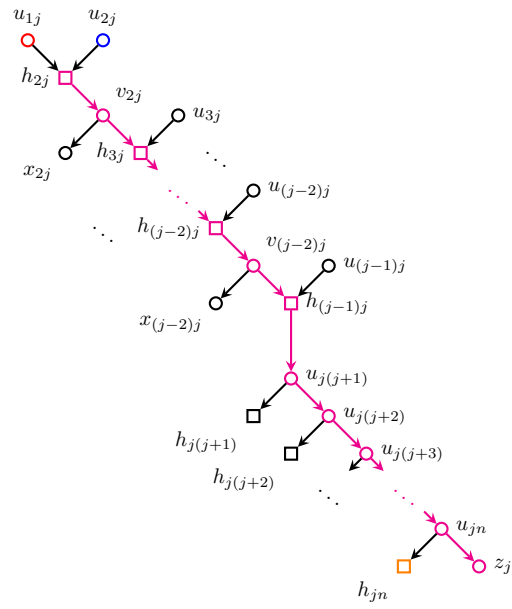

  It is trivial to verify that this is a tree-child phylogenetic network on the set of nodes labeled as $x_{ij}$ and $z_j$ and root $u_{12}$. This network $N$ consists of a single blob with exit reticulation $h_{(n-1)n}$ and its total number of reticulations is
  \begin{align*}
    \lvert \{ h_{ij} : 2 \leq i < j \leq n \} \rvert
      &= \sum_{j=3}^n \lvert \{ h_{ij} : 2 \leq i < j \} \rvert \\
      &= \sum_{j=3}^n (j - 2)
      = \sum_{j=1}^{n-2} j
      = \frac{(n-1)(n-2)}{2}.
  \end{align*}
  Moreover, removing all leaves of the form $x_{ij}$ and edges $v_{ij} x_{ij}$ incident to them, and contracting the pairwise disjoint paths (highlighted in Figures~\ref{fig:treeewidth-to-level-nodes12},~\ref{fig:treeewidth-to-level-nodej} and~\ref{fig:treeewidth-to-level-noden})
  \begin{align*}
    & u_{12} \cdots u_{1n} z_1 & \\
    & u_{23} \cdots u_{2n} z_2 & \\
    & h_{2j} \cdots h_{(j-1)j} u_{j(j+1)} \cdots u_{jn} z_j & \quad \text{for $j = 3,\ldots,n-1$} \\
    & h_{2n} \cdots h_{(n-1)n} z_n &
  \end{align*}  yields an underlying graph which is isomorphic to $K_n$. To prove this, observe that this yields a graph with $n$ nodes, namely $[z_i]$ for $i = 1, \ldots n$. Note that there are edges from each $[z_i]$ to all other $[z_j]$ for $2 \leq i < j \leq n$: these are witnessed by the edge $u_{ij} h_{ij}$ connecting $u_{ij} \in [z_i]$ and $h_{ij} \in [z_j]$. In addition, there are also edges from $[z_1]$ to each other $[z_j]$ for $2 < j \leq n$ via the edge $u_{1j} h_{2j}$ from $u_{1j} \in [z_1]$ to $h_{2j} \in [z_j]$. Finally, there is an edge from $[z_1]$ to $[z_2]$ via the edge $u_{12} u_{23}$.
\end{proof}
\FloatBarrier

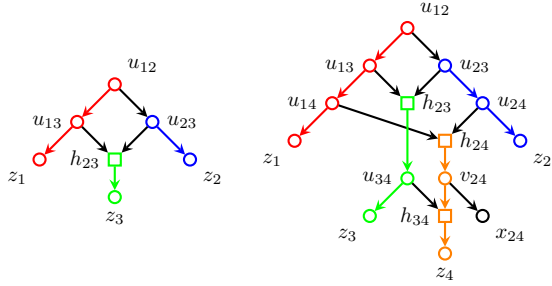
\begin{figure}[hpt]
  \centering
  \raisebox{-0.5\height}{
    \begin{tikzpicture}[scale=0.5, every node/.style={scale=0.8}, thick, >=stealth]
      \draw(0,0) node[trep, color=red, label=above right:${u_{12}}$] (u12) {};

      \path (u12)
        ++(-1, -1) node[trep, color=red, label=left:$u_{13}$]     (u13)  {}
        ++( 2,  0) node[trep, color=blue, label=right:$u_{23}$]    (u23) {}
        ++(-1, -1) node[hybp, color=green, label=left:$h_{23}$]     (h23)  {}
        ++( 2,  0) node[trep, color=blue, label=below right:$z_2$] (z2) {};

      \path (u13)
        ++(-1, -1) node[trep, color=red, label=below left:$z_1$] (z1) {};

      \path (h23)
        ++(0, -1) node[trep, color=green, label=below:$z_3$] (z3) {};

      \draw[->] (u12) -- (u23); \draw[->, color=red] (u12) -- (u13);
      \draw[->] (u13) -- (h23); \draw[->, color=red] (u13) -- (z1);
      \draw[->] (u23) -- (h23); \draw[->, color=blue] (u23) -- (z2);
      \draw[->, color=green] (h23) -- (z3);
    \end{tikzpicture}
  }
  \raisebox{-0.5\height}{
    \begin{tikzpicture}[scale=0.5, every node/.style={scale=0.8}, thick, >=stealth]
      \draw(0,0) node[trep, color=red, label=above right:${u_{12}}$] (u12) {};

      \path (u12)
        ++(-1, -1) node[trep, color=red, label=left:$u_{13}$]     (u13) {}
        ++( 2,  0) node[trep, color=blue, label=right:$u_{23}$]    (u23) {}
        ++(-1, -1) node[hybp, color=green, label=right:$h_{23}$]     (h23) {}
        ++( 2,  0) node[trep, color=blue, label=right:$u_{24}$]    (u24) {}
        ++(-1, -1) node[hybp, color=orange, label=right:$h_{24}$]    (h24) {}
        ++( 2,  0) node[trep, color=blue, label=below right:$z_2$] (z2) {};

      \path (u13)
        ++(-1, -1) node[trep, color=red, label=left:$u_{14}$] (u14) {}
        ++(-1, -1) node[trep, color=red, label=below left:$z_1$] (z1) {};

      \path (h23)
        ++( 0, -2) node[trep, color=green, label=left:$u_{34}$]    (u34) {}
        ++(-1, -1) node[trep, color=green, label=below left:$z_3$] (z3)  {}
        ++( 2,  0) node[hybp, color=orange, label=left:$h_{34}$]    (h34) {}
        ++( 0, -1) node[trep, color=orange, label=below:$z_4$] (z4)  {};

      \path (h24)
        ++(0, -1) node[trep, color=orange, label=right:$v_{24}$]       (v24) {}
        ++(1, -1) node[trep, label=below right:$x_{24}$] (x24) {};

      \draw[->] (u12) -- (u23); \draw[->, color=red] (u12) -- (u13);
      \draw[->] (u13) -- (h23); \draw[->, color=red] (u13) -- (u14);
      \draw[->] (u14) -- (h24); \draw[->, color=red] (u14) -- (z1);
      \draw[->] (u23) -- (h23); \draw[->, color=blue] (u23) -- (u24);
      \draw[->] (u24) -- (h24); \draw[->, color=blue] (u24) -- (z2);

      \draw[->, color=green] (h23) -- (u34); \draw[->, color=green] (u34) -- (z3);
      \draw[->] (u34) -- (h34);

      \draw[->, color=orange] (h24) -- (v24); \draw[->] (v24) -- (x24);
      \draw[->, color=orange] (v24) -- (h34); \draw[->, color=orange] (h34) -- (z4);
    \end{tikzpicture}
  }
  \caption{The networks constructed in Lemma~\ref{prop:treewidth-to-level} for $n=3$ and $n=4$.}
  \label{fig:treeewidth-to-level-examples}
\end{figure}

This is a necessary result: for it could (hypothetically) be the case that the topological properties of a tree-child network forbade its tree-width to attain high values (which is, for instance, what happens with trees: its tree-width is constant). However, we have proved that the tree-width of a tree-child network is in general unbounded.

\subsection{Tree-width and level of a tree-child network}\label{ssec:minors-tw}

In the previous section we showed that there was a tree-child phylogenetic network with level $O(n^2)$ such that its tree-width was at least $n$, for any $n \geq 3$. In this last section we shall travel in the opposite direction, and show that if a phylogenetic network contains a minor with a fixed number of nodes and edges, then it is possible to find a lower bound of its level.
This is useful to us, because a very widespread characterization of graph theoretical concepts (and indeed the tree-width) is done in terms of forbidden minors; i.e., minors whose presence in a graph determines whether it has a property or not. In the case of the tree-width, for instance, we know that a graph has tree-width less than $3$ if, and only if, it does not contain any minor isomorphic to $K_4$~\cite[p. 327]{diestel2000graph}.

\begin{lem}\label{lem:ind-step}
    Let $H$ be a weakly connected subgraph of a phylogenetic network, with $d$ nodes of in-degree $0$. Then, there are at least $d-1$ reticulation nodes in $H$.
\end{lem}
\begin{proof}
    We shall proceed by induction over the number of edges of $H$. If that number is $1$, then there is exactly one node with in-degree $0$ and there is nothing to prove.

    Assume now that the statement in this proposition is true for $H$ with up to $m$ edges, let $D_0(H)$ be the set of all nodes with in-degree $0$, and $d = |D_0(H)|$. Assume as well that $d\geq2$, for if $d = 1$ there is nothing to prove. Let $u$ be a node with in-degree $0$, and $h$ a reticulation node such that it descends from $u$ in $H$ and there is no other reticulation node between $u$ and $h$. Let $H'$ be the subgraph obtained from $H$ by deleting the path from $u$ to $h$ except for $h$; i.e., deleting all the nodes in that path different from $h$ and all the edges with either endpoint in them.
    Notice that this path is a strong path and all the nodes in it except $h$ is a tree node.

    Now, $H'$ need not be weakly connected. Let $H_1',\ldots, H_s'$ be the weakly connected components of $H'$, each with $d_i$ nodes of $D_0(H)$. Let $vw$ be a bridge edge from a node in the path $u\rightsquigarrow h$, say $v$, to a node in $H_i'$, say $w$; since $v$ is a tree node, the direction goes from $v$ to $w$. Two cases arise:
    \begin{itemize}
        \item If $w$ is a tree node in $H$, then $H_i'$ now has a new node with in-degree $0$. In this case, $H_i'$ will have at least $d_i+1$ nodes with in-degree $0$ and thus a minimum of $d_i$ reticulation nodes by the induction hypothesis.
        \item If $w$ is a reticulation node in $H$, it will not be so in $H_i'$. In this case $H_i'$ would have at least $d_i$ nodes with in-degree $0$ and therefore at least $d_i-1$ reticulation nodes ($d_i$ in $H$).
    \end{itemize}
    By definition, $d_1 + \cdots + d_s = d-1$ (because we have removed $u$ from $H$). Therefore, since each $H_i'$ contains at least $d_i$ reticulation nodes of $H$, we conclude that $H$ has at least $d-1$ reticulation nodes.
\end{proof}

Given a minor of a phylogenetic network $N$ (which is, remember, an undirected graph), we can endow it with an orientation induced by $N$ as follows: if $\alpha\beta$ (undirected) is an edge in such a minor, then there exist $u\in\alpha$ and $v\in\beta$ such that either $uv$ or $vu$ (directed) is an edge in $N$. Say, for instance, that $uv$ is the case; then, endow the original edge with the direction from $\alpha$ to $\beta$. Note that this does not prevent the existence of multiple orientations, but the following result is true nevertheless.

\begin{cor}\label{cor:ind-step}
    Let $\quot H\sim$ be a minor of a phylogenetic network $N$, endowed with an orientation induced by $N$. Let $\alpha_1,\ldots, \alpha_d,\,\beta$ be nodes in $\quot H\sim$ such that there exist edges $\alpha_1\beta, \ldots, \alpha_d\beta$. Then, there exist at least $d-1$ reticulation nodes in $\beta$.
\end{cor}
\begin{proof}
    By definition, there exist $u_1\in\alpha_1, \ldots,\, u_d\in\alpha_d$ and $v_1,\ldots,v_d\in\beta$ such that $u_1v_1,\ldots, u_dv_d$ are edges in $N$. Notice as well that the nodes in $\beta$ induce a weakly connected subgraph of $N$ (Section 1.7 in~\cite{diestel2000graph}). Consider this subgraph, together with the nodes $u_1,\ldots, u_d$ and the edges $u_1v_1,\ldots, u_dv_d$. We now fall under the hypotheses of Lemma~\ref{lem:ind-step} and deduce that there are at least $d-1$ reticulations in this graph, none of which can be $u_1,\ldots,u_d$ and thus must be in $\beta$.
\end{proof}

We  will now provide a lower bound of the level of a phylogenetic network given the number of nodes and edges of one of its minor. Since we have not been able to find a suitable reference in the literature, we present here the result with its proof.

\begin{prop}\label{prop:minor-to-level}
  Let $N$ be a phylogenetic network and $H$ any biconnected minor of $N$. Then the level of $N$ is at least $|E(H)| - |V(H)| + 1$.
\end{prop}
\begin{proof}
  Let $G$ be the underlying graph of $N$.
  If there is such a minor of $H$, it must be contained in a blob $B$ of $N$ (as they are maximal biconnected components by definition).
  In addition, recall that any minor of the underlying graph of $N$ is obtained by removing nodes, edges and contracting edges from $G$. This process is then described by the quotient of a subgraph $H'$ of $G$ via an equivalence relation $\sim$ (Proposition 1.7.1 in~\cite{diestel2000graph}). Then, the nodes of such a minor $H = \quot{H'}\sim$ are equivalence classes of $V(H')$, and there is an edge $[u][v] \in E(H)$ whenever there is an edge $u'v' \in E(H')$ or $v'u' \in E(H')$ for some $u' \in [u]$ and $v' \in [v]$. For each edge $[u][v]$ in $H$, choose such an edge in $H'$ in order to endow $[u][v]$ with an orientation and obtain a directed graph with underlying graph $H$. Moreover, one can choose some $r \in V(H')$ with no incoming edges from nodes in $H'$ (this is possible because $N$ is a DAG) and therefore there exists some orientation such that $\degin^H([r]) = 0$.

  By Corollary~\ref{cor:ind-step} we have that each node $\beta \in V(H)$ (i.e. equivalence class of weakly connected nodes in $H'$) in this directed graph with at least one incoming edge contains at least $\degin(\beta) - 1$ reticulation nodes of $N$. Summing over all such nodes $\beta \in V(H)$ we obtain a lower bound of the number of reticulations of the blob and thus the level of $N$:
  \begin{align*}
    |R(B)|
      &\geq \sum_{\substack{\beta \in V(H) \\ \degin(\beta) \geq 1}} (\degin(\beta) - 1) \\
      &= \sum_{\beta \in V(H)} \degin(\beta) - \lvert \{ \beta \in V(H) : \degin(\beta) \geq 1 \} \rvert \\
      &\geq |E(H)| - |V(H)| + 1.
  \end{align*}
\end{proof}

\begin{rem}
     Notice that the number $|E(H)| - |V(H)| + 1$ would be equal to the number of reticulations of $H$ if $H$ was a directed binary network. However, we have defined our minors to be undirected and not necessarily binary; in the binary case it is, then, the number of reticulations of a minor endowed with any orientation.
\end{rem}

\begin{rem}
  In the previous three results, we have thoroughly used the fact that our networks are binary.
\end{rem}

In particular, the following two results are useful to relate the level of a phylogenetic network to is tree-width. In Section~\ref{ssec:graph-theory} we saw that if a phylogenetic network had a minor isomorphic to a clique $K_n$, then its tree-width was at least $n-1$.

\begin{cor}
  If a phylogenetic network has a minor isomorphic to the clique $K_n$, then its level is at least $\frac{(n-1)(n-2)}{2}$.
\end{cor}

A similar result is given for the \emph{grid}. A $n\times m$ grid is the graph on the set of nodes $\{1,\ldots,n\}\times\{1,\ldots,m\}$ whose set of edges is
$$
\{(i,j)(i',j') : |i-i'| + |j-j'| = 1\}.
$$
A graph with a minor isomorphic to a $n\times m$ grid has tree-width at least $\min\{n,m\}$~\cite[p. 356]{diestel2000graph}. In the following corollary, we give a lower bound for a phylogenetic network with a minor isomorphic to a $n\times m$ grid.

\begin{cor}
  If a phylogenetic network has a minor isomorphic to the $n\times m$ grid or the $n\times m$ complete bipartite graph, then its level is at least $(n-1)(m-1)$.
\end{cor}

Given any forbidden minor of a particular value of the tree-width, Proposition~\ref{prop:minor-to-level} provides a way of finding a lower bound to the level a phylogenetic network in order for it to present that tree-width.

\section{Conclusions}

In this manuscript, we have reduced the problem of counting the number of cherry reduction sequences in tree-child networks to the well-studied problem of counting the number of linear extensions of their reachability order on tree nodes (Section~\ref{sec:lin-ext}), for which there exist several algorithms whose complexity depends on the tree-width of the underlying cover graph. In particular, the problem becomes $O(p^{\tw}$), where $p$ is the number of elements in the poset and $\tw$ is its tree-width. In this regard, we have shown (Section~\ref{ssec:arb-tw}) that the fact that our networks are tree-child does not imply the existence of a fixed upper bound on their tree-width. However, the tree-width of a phylogenetic network is bounded by its level. In this regard, we have given a lower bound of the level (and thus, of the tree-width) of a tree-child phylogenetic network  given the number of nodes and edges of a minor (Section~\ref{ssec:minors-tw}), which can be easily applied to any set of forbidden minors for a given tree-width.

This does not mean, however, that the problem of counting linear extensions in a tree-child network of a given level is as difficult as that of counting linear extensions in a general poset whose cover graph has the same level. The possibility remains open (at least, theoretically) for the existence of an algorithm specific to tree-child networks that exploits some topological features in order to find a quicker way to count their linear extensions. As future work, we would like to study what kind of posets do tree-child networks define and whether computing the number of linear extensions of these particular posets is a substantially easier problem than that of general posets.

On a more practical note, we want to underline that this connection also provides a way to polynomially approximate the number of cherry reductions by leveraging the results of~\cite{bubley1999faster} and~\cite[p.~2]{dyer1991random} on approximating the number of linear extensions of certain posets.
In conclusion, we hope that this manuscript can point the phylogenetist community to the tools developed to solve a larger problem.

\vskip1ex

\noindent \sloppy \textbf{Acknowledgements.} This research was partially supported by the Spanish Ministry of Economy and Competitiveness and European Regional Development Fund project PID2021-126114NB-C44 funded by MCIN/AEI/10.13039/501100011033 and by “ERDF A way of making Europe.”

\printbibliography

@book{Huson,
  title={Phylogenetic networks: concepts, algorithms and applications.},
  author={Huson, Daniel H and Rupp, Regula and Scornavacca, Celine},
  year={2010},
  publisher={Cambridge University Press}
}

@inproceedings{kangas2016counting,
  title={Counting Linear Extensions of Sparse Posets.},
  author={Kangas, Kustaa and Hankala, Teemu and Niinim{\"a}ki, Teppo Mikael and Koivisto, Mikko},
  booktitle={IJCAI},
  pages={603--609},
  year={2016}
}

@article{janssen2019treewidth,
  author = {Remie Janssen and Mark Jones and Steven Kelk and Georgios Stamoulis and Taoyang Wu},
  title = {Treewidth of display graphs: bounds, brambles and applications},
  journal = {Journal of Graph Algorithms and Applications},
  year = {2019},
  volume = {23},
  number = {4},
  pages = {715--743},
  doi = {10.7155/jgaa.00508}
}

@article{brightwell1991sharpcomplete,
  author = {Brightwell, Graham and Winkler, Peter},
  title = {Counting linear extensions},
  journal = {Order},
  year = {1991},
  volume = {8},
  number = {},
  pages = {225-242},
  doi = {10.1007/BF00383444}
}

@article{atkinson1990computing,
  title={On computing the number of linear extensions of a tree},
  author={Atkinson, Mike D},
  journal={Order},
  volume={7},
  number={1},
  pages={23--25},
  year={1990},
  publisher={Kluwer Academic Publishers Dordrecht}
}

@article{bubley1999faster,
  title={Faster random generation of linear extensions},
  author={Bubley, Russ and Dyer, Martin},
  journal={Discrete mathematics},
  volume={201},
  number={1-3},
  pages={81--88},
  year={1999},
  publisher={Elsevier}
}

@article{dyer1991random,
  title={A random polynomial-time algorithm for approximating the volume of convex bodies},
  author={Dyer, Martin and Frieze, Alan and Kannan, Ravi},
  journal={Journal of the ACM (JACM)},
  volume={38},
  number={1},
  pages={1--17},
  year={1991},
  publisher={ACM New York, NY, USA}
}

@article{felsner2015linear,
  title={Linear extensions of N-free orders},
  author={Felsner, Stefan and Manneville, Thibault},
  journal={Order},
  volume={32},
  number={2},
  pages={147--155},
  year={2015},
  publisher={Springer}
}

@article{habib1987some,
  title={On some complexity properties of N-free posets and posets with bounded decomposition diameter},
  author={Habib, Michel and M{\"o}hring, Rolf H},
  journal={Discrete Mathematics},
  volume={63},
  number={2-3},
  pages={157--182},
  year={1987},
  publisher={Elsevier}
}

@article{li2005computing,
  title={On computing the number of topological orderings of a directed acyclic graph},
  author={Li, Wing-Ning and Xiao, Zhichun and Beavers, Gordon},
  journal={Congressus Numerantium},
  volume={174},
  pages={143--159},
  year={2005}
}

@incollection{mohring1989computationally,
  title={Computationally tractable classes of ordered sets},
  author={M{\"o}hring, Rolf H},
  booktitle={Algorithms and order},
  pages={105--193},
  year={1989},
  publisher={Springer}
}

@article{peczarski2004new,
  title={New results in minimum-comparison sorting},
  author={Peczarski, Marcin},
  journal={Algorithmica},
  volume={40},
  pages={133--145},
  year={2004},
  publisher={Springer}
}

@article{eiben2019counting,
  title={Counting linear extensions: Parameterizations by treewidth},
  author={Eiben, Eduard and Ganian, Robert and Kangas, Kustaa and Ordyniak, Sebastian},
  journal={Algorithmica},
  volume={81},
  pages={1657--1683},
  year={2019},
  publisher={Springer}
}

@book{diestel2000graph,
  title={Graph theory},
  author={Diestel, Reinhard},
  year={2000},
  publisher={New York: Springer}
}

@article{cardona2008treechild,
  title={Comparison of tree-child phylogenetic networks},
  author={Cardona, Gabriel and Rossell{\'o}, Francesc and Valiente, Gabriel},
  journal={IEEE/ACM Transactions on Computational Biology and Bioinformatics},
  volume={6},
  number={4},
  pages={552--569},
  year={2008},
  publisher={IEEE}
}

@article{erdos2019orchard,
title = {A class of phylogenetic networks reconstructable from ancestral profiles},
journal = {Mathematical Biosciences},
volume = {313},
pages = {33-40},
year = {2019},
issn = {0025-5564},
doi = {https://doi.org/10.1016/j.mbs.2019.04.009},
url = {https://www.sciencedirect.com/science/article/pii/S0025556419300239},
author = {Péter L. Erdős and Charles Semple and Mike Steel}
}

@article{robertson1986treewidth,
title = {Graph minors. II. Algorithmic aspects of tree-width},
journal = {Journal of Algorithms},
volume = {7},
number = {3},
pages = {309-322},
year = {1986},
issn = {0196-6774},
doi = {https://doi.org/10.1016/0196-6774(86)90023-4},
url = {https://www.sciencedirect.com/science/article/pii/0196677486900234},
author = {Neil Robertson and P.D Seymour}
}

@article{kong2022classes,
  title={Classes of explicit phylogenetic networks and their biological and mathematical significance},
  author={Kong, Sungsik and Pons, Joan Carles and Kubatko, Laura and Wicke, Kristina},
  journal={Journal of Mathematical Biology},
  volume={84},
  number={6},
  pages={47},
  year={2022},
  DOI={10.1007/s00285-022-01746-y},
  publisher={Springer}
}

@article{van2022orchard,
  title={Orchard networks are trees with additional horizontal arcs},
  author={van Iersel, Leo and Janssen, Remie and Jones, Mark and Murakami, Yukihiro},
  journal={Bulletin of Mathematical Biology},
  volume={84},
  number={8},
  pages={76},
  year={2022},
  doi={10.1007/s11538-022-01037-z},
  publisher={Springer}
}

@article{Janssen2021,
   author = {Remie Janssen and Yukihiro Murakami},
   doi = {10.1016/j.tcs.2020.12.031},
   issn = {03043975},
   journal = {Theoretical Computer Science},
   month = {2},
   pages = {121-150},
   publisher = {Elsevier B.V.},
   title = {On cherry-picking and network containment},
   volume = {856},
   year = {2021},
}

@article{Cardona2009,
   author = {Gabriel Cardona and Francesc Rossello and Gabriel Valiente},
   doi = {10.1109/TCBB.2007.70270},
   issn = {1545-5963},
   issue = {4},
   journal = {IEEE/ACM Transactions on Computational Biology and Bioinformatics},
   month = {10},
   pages = {552-569},
   publisher = {IEEE},
   title = {Comparison of Tree-Child Phylogenetic Networks},
   volume = {6},
   year = {2009},
}

@article{bonnin2019situer,
  title={Situer l’analyse phylog{\'e}n{\'e}tique entre les sciences historiques et exp{\'e}rimentales},
  author={Bonnin, Thomas and Lombard, Jonathan},
  journal={Philosophia Scienti{\ae}. Travaux d'histoire et de philosophie des sciences},
  number={23-2},
  pages={131--148},
  year={2019},
  publisher={Universit{\'e} Nancy 2}
}

@book{semple2003phylogenetics,
  title={Phylogenetics},
  author={Semple, Charles and Steel, Mike and others},
  volume={24},
  year={2003},
  publisher={Oxford University Press on Demand}
}

@article{van2010locating,
  title={Locating a tree in a phylogenetic network},
  author={Van Iersel, Leo and Semple, Charles and Steel, Mike},
  journal={Information Processing Letters},
  volume={110},
  number={23},
  pages={1037--1043},
  year={2010},
  publisher={Elsevier}
}

@article{bordewich2016determining,
  title={Determining phylogenetic networks from inter-taxa distances},
  author={Bordewich, Magnus and Semple, Charles},
  journal={Journal of mathematical biology},
  volume={73},
  number={2},
  pages={283--303},
  year={2016},
  publisher={Springer}
}

@article{van2014trinets,
  title={Trinets encode tree-child and level-2 phylogenetic networks},
  author={Van Iersel, Leo and Moulton, Vincent},
  journal={Journal of mathematical biology},
  volume={68},
  number={7},
  pages={1707--1729},
  year={2014},
  publisher={Springer}
}

\end{document}